\documentclass[10pt,letterpaper]{article}

\usepackage{amsmath,amsfonts,graphpap,amscd,mathrsfs,graphicx,lscape}
\usepackage{epsfig,amssymb,amstext,xspace,bbm}
\usepackage{theorem}
\usepackage{algorithm,algorithmic}

\usepackage{color}              
\usepackage{epsfig}

\newcommand{\E}{\mathbb{E}}

\newtheorem{theorem}{Theorem}
\newtheorem{lemma}[theorem]{Lemma}

\newtheorem{corollary}[theorem]{Corollary}
\newtheorem{defn}[theorem]{Definition}

{\theorembodyfont{\rmfamily} }

\def\squareforqed{\hbox{\rule{2.5mm}{2.5mm}}}
\def\blackslug{\rule{2.5mm}{2.5mm}}
\def\qed{\hfill\blackslug}

\def\QED{\ifmmode\squareforqed 
  \else{\nobreak\hfil   
    \penalty50                 
    \hskip1em                  
    \null                      
    \nobreak                   
    \hfil                      
    \squareforqed              
    \parfillskip=0pt           
    \finalhyphendemerits=0     
    \endgraf}                  
  \fi}

\def\blksquare{\rule{2mm}{2mm}}
\def\qedsymbol{\blksquare}
\newcommand{\bg}[1]{\medskip\noindent{\bf #1}}
\newcommand{\ed}{{\hfill\qedsymbol}\medskip}

\newenvironment{proof}{\noindent\textbf{Proof.}}{\ed}

\newcommand{\comment}[1]{}

\newcommand{\R}{\ensuremath{\mathbb R}}

\newcommand{\A}{\ensuremath{\mathcal{A}}}
\newcommand{\C}{\ensuremath{\mathcal{C}}}

\newcommand{\M}{\ensuremath{\mathcal M}}

\newcommand{\opt}{\text{\textsc{Opt}} }

\begin{document}

\title{Bayesian Games and the Smoothness Framework}

\author{Vasilis Syrgkanis}

\maketitle

\begin{abstract}
We consider a general class of Bayesian Games where each players utility
depends on his type (possibly multidimensional) and on the strategy profile and where players' 
types are distributed independently. We show that if their 
full information version for any fixed instance of the type profile is a smooth game then the Price of Anarchy bound
implied by the smoothness property, carries over to the Bayes-Nash Price of Anarchy. 
We show how some proofs from the literature (item bidding auctions, greedy auctions) can be 
cast as smoothness proofs or be simplified using smoothness. For first price item 
bidding with fractionally subadditive bidders we actually manage to improve by much the existing 
result \cite{Hassidim2011a} from $4$ to $\frac{e}{e-1}\approx 1.58$. This also shows a very 
interesting separation between first and second price item bidding since second price item bidding
has PoA at least $2$ even under complete information. For a larger class of 
Bayesian Games where the strategy space of a player also changes with his type we are able to show
that a slightly stronger definition of smoothness also implies a Bayes-Nash PoA bound.
We show how weighted congestion games actually satisfy this stronger definition of smoothness.
This allows us to show that the inefficiency bounds of weighted congestion games  
known in the literature carry over to incomplete versions where the weights
of the players are private information. We also show how an incomplete version of
a natural class of monotone valid utility games, called effort market games are 
universally $(1,1)$-smooth. Hence, we show that incomplete versions of effort market
games where the abilities of the players and their budgets are private information 
has Bayes-Nash PoA at most $2$. 
\end{abstract}

\section{Introduction}
In our information era, with the advent of electronic markets, most systems have grown
so large scale that central coordination has become infeasible. In addition players
are less and less informed of the actual game they are playing and the type of players
they are competing against. Coordinating the players is too costly and assuming that
the players know all the parameters of the game is too simplistic. Such a realization
renders mandatory the study of efficiency in non-cooperative games of incomplete information. 

Ever since the introduction of the concept of the Price of Anarchy, a large part
of the algorithmic game theory literature has studied the effects of selfishness in
the efficiency of a system. In a unifying paper, Roughgarden \cite{Roughgarden2009}, gave a 
general technique, called smoothness, for proving inefficiency results in games and portrayed how many 
of the results in the literature can be cast in his framework. In addition, he showed 
that such types of inefficiency proofs extend directly to almost every reasonable 
non-cooperative solution concept, such as pure nash equilibria, mixed nash equilibria,
correlated equilibria and coarse-correlated equilibria. 

However, such a unification holds only under the strong assumption of complete information
\footnote{Recently we became aware that an independent work by Roughgarden with some overlapping results on
extending smoothness to incomplete information games has been under submission to a conference since February 6, 2012.}: 
players know every parameter of the game and no player has private information. Such an 
assumption is pretty strong and if we want models that could capture realistic environments 
we need to cope with games were players have incomplete information. 

In this work we manage to show how to extend this unification to a significant class of games of 
incomplete information: basically we show that if a complete information game is 
smooth then the inefficiency bound given by the smoothness argument carries over to 
incomplete information versions of the game where players have private parameters and each player's utility depends on his parameter and
the actions of the rest of the players. Hence, we manage to unify Price of Anarchy with 
Bayesian Price of Anarchy analysis for a large class of games. 

Many of the games studied in the literature such as Weighted Congestion Games 
have been shown to be tight games: the best efficiency bound provable with a smoothness argument is the best
bound possible. For such games our analysis shows that the bayesian price of anarchy bound
is also tight, since complete information is a special case of incomplete information and hence
lower bounds on inefficiency carry over to the incomplete information model. Hence, 
this immediately implies that the efficiency guarrantees for a tight class of games don't depend
on the information pocessed by the players: having more information doesn't imply better efficiency
and having less information doen't imply worse efficiency. 

Our approach also manages to put several Bayesian Price of Anarchy results that exist in 
the literature under the smoothness framework. Specifically we show how our main theorem
can be applied to get an improved result for first price item-bidding with subadditive bidders
and how to get a much simpler proof of good approximation guarrantees of the greedy mechanisms
introduced by Lucier and Borodin \cite{Lucier2009}. 

\subsection{Related Work}
There has been a long line of research on quantifying inefficiency of equilibria starting
from \cite{Koutsoupias1999} who introduced the notion of the price of anarchy. A recent work
by Roughgarden \cite{Roughgarden2009} managed to unify several of these results under a proof
framework called smoothness and also showed that such inefficiency proofs also carry over
to inefficiency of coarse correlated equilibria. Moreover, he showed that such techniques 
give tight results for the well-studied class of congestion games. Later, Bhawalkar et al. 
\cite{Bhawalkar2010} also showed that it produces tight results for the larger class of
weighted congestion games. Another recent work by Schoppman and Roughgarden \cite{Roughgarden2010}
copes with games with continuous strategy spaces and shows how the smoothness framework should
be adapted for such games to produce tighter results. The introduce the new notion of local
smoothness for such games and showed that if an inefficiency upper bound proof lies in this framework 
then it also carries over to correlated equilibria. 

There have also been several works on quantifying the inefficiency of incomplete information games,
mainly in the context of auctions. A series of papers by Paes Leme and Tardos \cite{PaesLeme2010}, 
Lucier and Paes Leme \cite{Lucier2011} and Caragiannis et al \cite{Caragiannis2011} studied the
ineffficiency of Bayes-Nash equilibria of the generalized second price auction. Lucier and Borodin 
studied Bayes-Nash Equilibria of non-truthful auctions that are based on greedy allocation algorithms
\cite{Lucier2009}. A series of three papers, Christodoulou, Kovacs and Schapira \cite{Christodoulou2008},
Bhawalkar and Roughgarden \cite{Bhawalkar} and Hassidim, Kaplan, Mansour, Nisan \cite{Hassidim2011a},
studied the inefficiency of Bayes-Nash equilibria of non-truthful combinatorial auctions that
are based on running simultaneous separate item auctions for each item. However, many of the results
in this line of work where specific to the context and a unifying framework doesn't exist. Lucier
and Paes Leme \cite{Lucier2011} introduced the concept of semi-smoothness and showed that their proof
for the inefficiency of the generalized second price auction falls into this category. However, 
semi-smoothness is a much more restrictive notion of smoothness than just requiring that every complete
information instance of the game to be smooth. 

\subsection{Model and Notation}

We consider the following model of Bayesian Games: Each player
has a type space $T_i$ and a probability distribution $D_i$ defined on $T_i$. 
The distributions $D_i$ are independent, the types $T_i$ are disjoint and we denote with $D=\times_i D_i$.
Each player has a set of actions $A_i$ and let $A=\times_i A_i$. The utility of a player is a function
$u_i: T_i\times A\rightarrow \R$. The strategy of each player is a function
$s_i: T_i\rightarrow A_i$. At times we will use the notation $s(t)=(s_i(t_i))_{i\in N}$
to denote the vector of actions given a type profile $t$ and $s_{-i}(t_{-i})=(s_j(t_j))_{j\neq i}$
to denote the vector of actions for all players except $i$. We could also define 
cost minimization games where each player has a cost $c_i:T_i\times A\rightarrow \R$. 
All of our results hold for both utility maximization and cost minimization games.

The two basic assumptions that we made in the class of Bayesian Games that 
we examine is that a players utility is affected by the other players'
types only implicitly through their actions and not directly from
their types and that the players' types are distributed independently. 
The above class of games is general enough and we portray several nice 
examples of such Bayesian Games. An interesting future direction
is to try and relax any of these two assumptions or show that smoothness
is not sufficient to prove bounds without these assumptions.

As our solution concept we will use the most dominant solution concept in
incomplete information games, the Bayes-Nash Equilibrium (BNE). Our results hold for
mixed Bayes-Nash Equilibria too, but for simplicity of presentation we are 
going to focus on Bayes-Nash Equilibria in pure strategies. A Bayes-Nash
Equilibrium is a strategy profile such that each player maximizes his expected
utility conditional on his private information:
$$\forall a \in A_i:\E_{t_{-i}|t_i}[u_i^{t_i}(s(t))]\geq \E_{t_{-i}|t_i}[u_i^{t_i}(a,s_{-i}(t_{-i})]$$
Given a strategy profile $s$ the social welfare of the game is defined
as the expected sum of player utilities: 
$$SW(s)=\E_t[SW^t(s(t))]=\E_t[\sum_i u_i^{t_i}(s(t))]$$
In addition given a type profile $t$ we denote with $\opt(t)$ the
action profile that achieves maximum social welfare for type profile $t$:
$\opt(t)=\arg\max_{a\in A}SW^t(a)$. 

As our measure of inefficiency we will use the \textit{Bayes-Nash Price of Anarchy} which 
is defined as the ratio of the expected optimal social welfare over the expected
social welfare achieved at the worst Bayes-Nash Equilibrium:
$$\sup_{s \text{ is } BNE}\frac{\E_{t}[SW(\opt(t))]}{\E_t[SW(s(t))]}$$

\section{Constant Strategy Space Bayesian Games and Smoothness}

In this section we give a general theorem on how one can derive Bayes-Nash Price of Anarchy
results using the smoothness framework introduced in \cite{Roughgarden2009}. We first
give a definition of smoothness suitable for incomplete information games. Our definition
just states that the game is smooth in the sense of \cite{Roughgarden2009} for 
each instantiation of the type profile.

\begin{defn}
A bayesian utility maximization game is said to be $(\lambda,\mu)$-smooth
if for any $t\in T$ and for any pair of action profiles $a,a'\in A$:
$$\sum_i u_i^{t_i}(a_i',a_{-i})\geq \lambda SW^t(a')-\mu SW^t(a)$$ 
\end{defn}

For the case of complete information games (i.e. the set of possible type
profiles is a singleton) Roughgarden \cite{Roughgarden2009} showed that the 
efficiency achieved by any coarse-correlated equilibrium (superset of pure nash,
mixed nash and correlated equilibrium) is at least a $\lambda/(1+\mu)$ fraction
of the optimal social welfare. In other words the price of anarchy of any
of these solution concepts is at most $(1+\mu)/\lambda$. Our main results shows
that the latter also extends to Bayes-Nash Equilibria for the case of incomplete
information. 

\begin{theorem}[Main Theorem]\label{thm:main}
 If a \textit{Bayesian Game} is $(\lambda,\mu)$-smooth then
it has Bayes-Nash Price of Anarchy at most $\frac{1+\mu}{\lambda}$.
\end{theorem}
\begin{proof}
 Let $\opt(t)=(\opt_i(t))_{i\in N}$ be the action profile that maximizes 
social welfare given a type profile $t$. Suppose that player $i$ with type $t_i$ switches to playing
$\opt_i(t_i,w_{-i})$ for some type profile $w_{-i}$. Let $s$ be a Bayes-Nash 
Equilibrium. Then we have:
\begin{align*}
\E_{t_{-i}}[u_i^{t_i}(s(t))]
  ~\geq~ &  \E_{t_{-i}}[u_i^{t_i}(\opt_i(t_i,w_{-i}), s_{-i}(t_{-i}))]
\end{align*}
Taking expectation over $t_i$ and over all possible type profiles $w_{-i}$ we get:
\begin{align*}
 \E_{t}[u_i^{t_i}(s(t))]
 ~\geq~ & \E_{w_{-i}}\E_{t_i}\E_{t_{-i}}[u_i^{t_i}(\opt_i(t_i,w_{-i}), s_{-i}(t_{-i}))] \\
  ~=~& \E_{w_{-i}} \E_{w_i} \E_{t_{-i}}[u_i^{w_i}(\opt_i(w_i,w_{-i}), s_{-i}(t_{-i}))]  \\
  ~=~& \E_{t_i} \E_{w_{-i}} \E_{w_i} \E_{t_{-i}}[u_i^{w_i}(\opt_i(w_i,w_{-i}), s_{-i}(t_{-i}))]  \\
  ~=~& \E_{t} \E_{w} [u_i^{w_i}(\opt_i(w), s_{-i}(t_{-i}))] \\
\end{align*}
 Adding the above inequality for all players and using the smoothness property we get:
\begin{align}
\E_t[SW^t(s(t))]=\sum_i \E_t[u_i^{t_i}(s(t))]
 ~\geq~ & \sum_i\E_{t} \E_{w} [u_i^{w_i}(\opt_i(w), s_{-i}(t_{-i}))] \nonumber \\
 ~=~ & \E_{t} \E_{w} [\sum_i u_i^{w_i}(\opt_i(w), s_{-i}(t_{-i}))] \nonumber \\
 ~\geq~& \E_{t} \E_{w} [\lambda SW^{w}(\opt(w)) - \mu SW^w(s(t))] \nonumber \\
 ~=~& \lambda \E_w[SW^{w}(\opt(w))] - \mu \sum_i \E_{t} \E_{w} [u_i^{w_i}(s_i(t_i),s_{-i}(t_{-i}))] \label{eqn:bayes-nash}
\end{align}
If in the pre-last line we had $SW^t(s(t))$ then we would directly get our result. However,
the fact that there is this misalignment we need more work. In fact we are going to prove that:
\begin{equation}\label{eqn:misalignment}
 \E_{t} \E_{w}[SW^w(s(t))]\leq \E_t[SW^t(s(t))]
\end{equation}

To achieve this we are going to use again the Bayes-Nash Equilibrium definition. But now
we are going to use it in the following sense: no player $i$ of some type $w_i$ 
wants to deviate to playing as if he was some other type $t_i$. This translates to:
\begin{align*}
 \E_{t_{-i}}[u_i^{w_i}(s_i(t_i),s_{-i}(t_{-i})] \leq \E_{t_{-i}}[u_i^{w_i}(s_i(w_i),s_{-i}(t_{-i})]
\end{align*}
Taking expectation over $t_i$ and $w$ we have:
\begin{align*}
 \E_w\E_t[u_i^{w_i}(s_i(t_i),s_{-i}(t_{-i})] ~\leq~ & \E_{w_i}\E_{w_{-i}}\E_{t_i}\E_{t_{-i}}[u_i^{w_i}(s_i(w_i),s_{-i}(t_{-i})]\\
 ~=~& \E_{w_i}\E_{t_{-i}}[u_i^{w_i}(s_i(w_i),s_{-i}(t_{-i})] \\
 ~=~& \E_{t_i}\E_{t_{-i}}[u_i^{t_i}(s_i(t_i),s_{-i}(t_{-i})]\\
 ~=~& \E_t[u_i^{t_i}(s(t))]
\end{align*}
Summing over all players gives us inequality \ref{eqn:misalignment}. Now combining inequality \ref{eqn:misalignment}
with inequality \ref{eqn:bayes-nash} we get:
\begin{equation}
 \E_t[SW^t(s(t))]\geq \lambda \E_w[SW^{w}(\opt(w))] - \mu \E_t[SW^t(s(t))]
\end{equation}
which gives the theorem.
\end{proof}

In fact it is easy to see that our above analysis also works for a more relaxed version of smoothness 
similar to the variant introduced in \cite{Lucier2011}:
 
\begin{defn}\label{def:semi-smooth}
A Bayesian utility maximization game is said to be $(\lambda,\mu)$-smooth
if for any $t\in T$ and $a\in A$, there exists a strategy profile $a'(t)$
such that:
$$\sum_i u_i^{t_i}(a_i'(t),a_{-i})\geq \lambda SW(\opt(t))-\mu SW^t(a)$$ 
\end{defn}

In fact our main proof holds even for a slightly more relaxed smoothness property that will 
prove useful in auction settings. Our main theorem works for the following notion
of smoothness under the condition that the utilities of players at equilibrium are non-negative.
\begin{defn}
 A Bayesian utility maximization game is said to be $(\lambda,\mu)$-smooth
if for any $t\in T$ and $a\in A$, there exists a strategy profile $a'(t)$
such that:
$$\sum_i u_i^{t_i}(a_i'(t),a_{-i})\geq \lambda SW(\opt(t))-\mu \sum_{i\in K\subseteq [n]}u_i^t(a)$$ 
where $K$ is some fixed subset of the players independent of the type profile $t$ and bid profile $b$.
\end{defn}

The reason why the latter definition is useful is that some auction environments might fail
to be smooth with the first definition mainly due to the fact that the utility of players
might be non-negative in expectation at equilibrium but can certainly be negative if we consider
an arbitrary bid profile $b$ that is not in equilibrium. Hence, the latter helps in settings 
where at equilibrium utilities are certainly non-negative (individual rationality) but 
there are strategy profiles at which a player might be getting negative utility.

\section{Item-Bidding Auctions}
In this section we consider the Item-Bidding Auctions studied in Christodoulou et al \cite{Christodoulou2008},
Bhawalkar and Roughgarden \cite{Bhawalkar} and Hassidim et al. \cite{Hassidim2011a}. 

We first prove a smoothness result for First Price Item-Bidding Auctions for fractionally
subadditive bidders. Fractionally subadditive bidders are subcase of additive bidders and
a generalization of submodular bidders. Our results imply a big improvement in existing results.
Specifically Hassidim et al show that for fractionally subadditive bidders the Bayes-Nash
Price of Anarchy is at most $4$ (and at most $4\beta$ for $\beta$-fractionally subadditive. 
We show that it is at most $\frac{e}{e-1}\approx 1.58$ (and at most $\frac{e}{e-1}\beta$ correspondingly). 

Thus our result gives the same guarantees for first price item bidding auctions as those
existing for second price item bidding auctions (e.g. Christodoulou et al \cite{Christodoulou2008},
Bhawalkar et al \cite{Bhawalkar}). 

\begin{theorem}
First Price Item-Bidding Auctions are $(\frac{1}{2},0)$-semi-smooth for 
fractionally subadditive bidders.
\end{theorem}
\begin{proof}
Consider a valuation profile $v$ and a bid profile $b$. Let $\opt(v)$ be the optimal
allocation for that type profile. Let $\opt_i(v)$ be the set that player 
$i$ gets at $\opt(v)$. Let $a=(a_1,\ldots,a_m)$
be the maximizing additive valuation for player $i$ for $\opt_i(v)$, i.e.
$v_i(\opt_i(v))=\sum_{j \in \opt_i(v)}a_j$ and $\forall S\neq \opt_i(v): v_i(S)\geq \sum_{j\in S}a_j$
(Such an $a$ exists by the definition of fractionally subadditive valuations).
Suppose that player $i$ switches to bidding $a_j/2$ for each $j\in \opt_i(v)$
and $0$ everywhere else. Denote with $b_i'(v)$ such a deviation.  
Let $X_i$ be the items that he wins after the deviation.
This means that for all $j\in \opt_i(v)-X_i: p_j(b)\geq a_j/2$ and for all $j\in X_i$
player $i$ pays exactly $a_j/2$.
Thus we have:
\begin{align*}
u_i(b_i'(v),b_{-i})\geq &~ v_i(X_i)-\sum_{j\in X_i}\frac{a_j}{2}
\geq  \sum_{j\in X_i}a_j -\sum_{j\in X_i}\frac{a_j}{2}
= \sum_{j\in X_i}\frac{a_j}{2}\\
\geq &~ \sum_{j\in X_i}\frac{a_j}{2} + \sum_{j\in \opt_i(v)-X_i}\frac{a_j}{2}-p_j(b)\\
\geq &~ \sum_{j\in \opt_i(v)}\frac{a_j}{2}-\sum_{j\in \opt_i(v)-X_i}p_j(b)\\
\geq &~ \sum_{j\in \opt_i(v)}\frac{a_j}{2}-\sum_{j\in \opt_i(v)}p_j(b)\\
= &~ \frac{v_i(\opt_i(v))}{2}-\sum_{j\in \opt_i(v)}p_j(b)
\end{align*}

Now summing over all players the second term on the right hand side will become 
the sum of prices over all items, since the sets $\opt_i(v)$ are disjoint. Thus we get:
\begin{equation}\label{eqn:smooth-first-price}
\sum_i u_i^{v_i}(b_i'(v),b_{-i}) + \sum_{j\in [m]}p_j(b)\geq \frac{\sum_i v_i(\opt_i(v))}{2}= \frac{1}{2}SW^v(\opt(v))
\end{equation}

Now the above inequality gives us the smoothness property. To completely fit it in our smoothness
model, we should also view the seller as a player with only one strategy and only one type and
whose utility is the sum of payments. His optimal deviation is then the trivial of doing nothing. 
Then the left hand side of the above inequality is the sum of the utilities of all the players 
(including the seller) had each of them unilaterally deviated to their optimal strategy. 
\end{proof}

In fact considering randomized deviations, similar to that of \cite{Lucier2011} we are able
to prove a much tighter bound of $\frac{e}{e-1}\approx 1.58$ on the Price of Anarchy by
showing that the above game is actually $(1-\frac{1}{e},0)$-semi-smooth.

\begin{theorem}
First Price Item-Bidding Auctions are $(1-\frac{1}{e},0)$-semi-smooth for 
fractionally subadditive bidders.
\end{theorem}
\begin{proof}
Consider a valuation profile $v$ and a bid profile $b$. Let $\opt(v)$ be the optimal
allocation for that type profile. Let $\opt_i(v)$ be the set that player 
$i$ gets at $\opt(v)$. Let $a=(a_1,\ldots,a_m)$
be the maximizing additive valuation for player $i$ for $\opt_i(v)$, i.e.
$v_i(\opt_i(v))=\sum_{j \in \opt_i(v)}a_j$ and $\forall S\neq \opt_i(v): v_i(S)\geq \sum_{j\in S}a_j$
(Such an $a$ exists by the definition of fractionally subadditive valuations).
Suppose that player $i$ switches to bidding a randomized bid with probability
density function $f(b)=\frac{1}{a_j-b}$ for $b\in [0,a_j(1-\frac{1}{e})]$ for each $j\in \opt_i(v)$
and $0$ everywhere else. Randomization is independent for each $j$. Denote with $\tilde{B}_i(v)$ such a 
randomized deviation and $\tilde{b}_i$ a random draw from $\tilde{B}_i(v)$.   
Let $X_i(b_i)$ be the random variable that denotes the items that he wins after the deviation.
Thus we have:
\begin{align*}
u_i(\tilde{B}_i(v),b_{-i})\geq &~ \E_{\tilde{b}_i\sim \tilde{B}_i(v)}[v_i(X_i(\tilde{b}_i))-\sum_{j\in X_i(\tilde{b}_i)}\tilde{b}_{ij}] \\
\geq~& \E_{\tilde{b}_i\sim \tilde{B}_i(v)}[\sum_{j\in X_i(\tilde{b}_i)}a_j -\sum_{j\in X_i(\tilde{b}_i)}\tilde{b}_{ij}]\\
\geq~& \E_{\tilde{b}_i\sim \tilde{B}_i(v)}[\sum_{j\in \opt_i(v)}(a_j -\tilde{b}_{ij})\mathbbm{1}\{\tilde{b}_{ij}\geq p_j(b)\}]\\
=~& \sum_{j\in \opt_i(v)}\E_{\tilde{b}_i\sim \tilde{B}_i(v)}[(a_j -\tilde{b}_{ij})\mathbbm{1}\{\tilde{b}_{ij}\geq p_j(b)\}]\\
=~& \sum_{j\in \opt_i(v)}\int_{p_j(b)}^{a_j(1-\frac{1}{e})}(a_j -t)\frac{1}{a_j-t}dt\\
=~& \sum_{j\in \opt_i(v)}a_j\left(1-\frac{1}{e}\right)-p_j(b)\\
=~& \left(1-\frac{1}{e}\right)v_i(\opt_i(v))-\sum_{j\in \opt_i(v)}p_j(b)
\end{align*}

Now summing over all players the second term on the right hand side will become 
the sum of prices over all items, since the sets $\opt_i(v)$ are disjoint. Thus we get:
\begin{equation}
\sum_i u_i^{v_i}(b_i'(v),b_{-i}) + \sum_{j\in [m]}p_j(b)\geq \left(1-\frac{1}{e}\right)\sum_i v_i(\opt_i(v))=\left(1-\frac{1}{e}\right)SW^v(\opt(v))
\end{equation}
\end{proof}

Similarly one can also show that for $\beta$-fractionally subadditive bidders the First Price
Item-Bidding Auction is $(\frac{1}{\beta}\left(1-\frac{1}{e}\right),0)$-semi-smooth.

\begin{corollary}
First Price Item-Bidding Auctions with independent $\beta$-fractionally subadditive bidders have Bayes-Nash Price
of Anarchy at most $\frac{e}{e-1}\beta$. 
\end{corollary}

It is known (see \cite{Bhawalkar}) that subadditive bidders are $\ln m$-fractionally subadditive. Thus the latter 
gives a price of anarchy of $\frac{e}{e-1}\ln m $ for subadditive bidders. 

The latter result creates an interesting separation between first-price and second price
auctions even in the incomplete information case. For the complete information case it is known
that second price item bidding has price of anarchy at least $2$ \cite{Christodoulou2008}, whilst pure nash equilibria of 
first price auctions are always optimal (when a pure nash exists) \cite{Hassidim2011a}. The above bound states that
such a separation even in the incomplete information case, since second price has price of anarchy
at least $2$ whilst first price auctions have price of anarchy at most $\approx 1.58$.

\section{Greedy First Price Auctions}

In this section we consider the greedy first price auctions introduced by Lucier and Borodin \cite{Lucier2009}.
In a greedy auction setting there are $n$ bidders and $m$ items. Each bidder $i$ have some private combinatorial
valuation $v_i$ on the items that is drawn from some commonly known distribution. 
The strategies of the players is to submit a valuation profile $b_i$ that outputs a value $b_i(S)$
for each set of items $S$. An allocation is a vector $A=(A_1,\ldots,A_n)$ that allocates 
a set $A_i$ for each player $i$. We assume that there is a predefined subspace of feasible allocations. 
The above setting is a generalization of the combinatorial auction setting since we don't assume
that the allocations $A_i$ must be disjoint, i.e. the same item can potentially be allocated to more
than one players.

A mechanism $\M(b)=(\A(b),p(b))$ takes as input a bid profile $b$ and outputs a feasible allocation 
$\A(b)$ and a vector of prices $p(b)$ that each player has to pay. A mechanism is said to be greedy
if the allocation output by the mechanism is the outcome of a greedy algorithm 
as we explain. Given bid profile $b$ a greedy algorithm is defined as follows: 
Let $r:[n]\times 2^{[m]}\times \R \rightarrow \R$ be a priority function such that
$r(i,S,v)$ is the priority of allocating set $S$ to player $i$ when $b_i(S)=v$. Then the greedy
algorithm is as follows: Pick allocation $(i,A_i)$ that maximizes $r(i,S,b_i(S))$ over all
currently feasible allocations $(i,S)$ and allocate set $A_i$ to player $i$. Then remove player $i$. 
A greedy algorithm is $c$-approximate if for any bid profile $b$ it returns an allocation that
is at least a $c$-fraction of the maximum possible allocation given valuation profile $b$.  
Moreover, the mechanism is first price if the payments that the mechanism outputs is
$p_i(b)=b_i(A_i)$.

Given a type profile $v$ and a bid strategy profile $b$ the social welfare is
as always the sum of the players' utilities (including the auctioneer as a player). This boils down 
to being the value of the allocation $\sum_i v_i(A_i)$ since payments cancel out. 

The game defined by a greedy mechanism is a Separable Bayesian Game and in the theorem that follows 
we are able to prove that it is also $(\frac{1}{2},c-1)$ smooth. This leads to a price of anarchy of $2c$. 
This is not an improvement to the existing result of $c+O(\log(c))$ by Lucier and Borodin, but the analysis
is much simpler and the difference in the two bounds is not big. 

\begin{theorem}
 Any Greedy First Price Mechanism based on a $c$-approximate greedy algorithm defines a
Bayesian Game that is $(\frac{1}{2},c-1)$-smooth.
\end{theorem}
\begin{proof}
 Given the mechanism $\M$ we define as $\theta_i(S,b_{-i})$ as the critical value 
that player $i$ has to bid on set $S$ such that he is allocated set $S$ given the bid
profiles of the rest of the players. 

In our proof we will use a very nice fact about $c$-approximate greedy mechanisms that 
was proved by Lucier and Borodin. For any $c$-approximate greedy mechanism and any feasible 
allocation $A'$ it must hold that $\sum_i \theta_i(A_i',b_{-i})\leq c\sum_i b_i(A_i)=c\sum_i p_i(b)$.

Now consider a valuation profile $v$ and bid profile $b$. Let $\opt_i(v)$ be the set allocated
to player $i$ in the optimal allocation for $v$. Let $b_i'(v)$ be the following bid strategy for player $i$:
he bids $v_i(\opt_i(v))/2$ single-mindedly on $\opt_i(v)$, i.e. $\forall S\neq \opt_i(v): b_i(S)=0$
and $b_i(\opt_i(v))=\frac{v_i(\opt_i(v))}{2}$. There are two cases: either he gets allocated his
optimal item in which case his $u_i(b_i'(v),b_{-i})=\frac{v_i(\opt_i(v))}{2}$ or he doesn't 
in which case: $\theta_i(\opt_i(v),b_{-i})\geq \frac{v_i(\opt_i(v))}{2}$. Therefore, we get:
\begin{equation}
 u_i(b_i'(v),b_{-i})\geq \frac{v_i(\opt_i(v))}{2}-\theta_i(\opt_i(v),b_{-i})
\end{equation}
Summing over all players and using the upper bound on critical price proved by Lucier and Borodin instantiated
for $A'=\opt(v)$ we get:
\begin{equation}
 \sum_iu_i(b_i'(v),b_{-i}) \geq \sum_i \frac{v_i(\opt_i(v))}{2} - \sum_i \theta_i(\opt_i(v),b_{-i})
\geq \frac{1}{2}\sum_i v_i(\opt_i(v)) - c \sum_i p_i(b)
\end{equation}
Now, by rearranging we get:
\begin{equation}
 \sum_iu_i(b_i'(v),b_{-i}) + \sum_i p_i(b) \geq \sum_i \frac{v_i(\opt_i(v))}{2} - \sum_i \theta_i(\opt_i(v),b_{-i})
\geq \frac{1}{2}\sum_i v_i(\opt_i(v)) - (c-1) \sum_i p_i(b)
\end{equation}

The latter states that our game satisfies our most relaxed semi-smoothness definition for $\lambda=\frac{1}{2}$
and $\mu=c-1$.
\end{proof}

In fact, we can improve our bound on the price of anarchy to $\frac{e}{e-1}c$ using the
same trick of randomized deviations that we did in item bidding. 

\section{Variable Strategy Space Games and Universal Smoothness}

In this section we cope with the following more general class of Bayesian Games
whose goal is to capture games where the strategy space of a player not only his 
utility is dependend on his type. We denote with $A_i(t_i)\subseteq A_i$ the actions available
to a player of type $t_i$. A players strategy is a function $s_i:T_i\rightarrow A_i$ that 
satisfies $\forall t_i\in T_i: s_i(t_i)\in A_i(t_i)$. We will denote with $A(t)=\times_i A_i(t_i)$.
We still assume that the utility of a player depends on his type and the actions of the 
rest of the players: $u_i:T_i\times A\rightarrow \R$. We must point out that the utility of a player $i$ is undefined
if $a_i\notin A_i(t_i)$. The rest of the definitions are
the same as in our previous definition of Bayesian Games. Our initial definition was
a special case of this class of Bayesian Games where $\forall t_i\in T_i:A_i(t_i)=A_i$.

For such a Bayesian Game we need a slight alteration of the definition of what it means
for a complete information instance of it to be smooth, since the utility of the complete
information instance is undefined on strategies that are not in the strategy space of a player
for that instance. 
\begin{defn}
 A Bayesian Game is $(\lambda,\mu)$-smooth if $\forall t\in T$ and for all $a,a'\in A(t)$:
$$\sum_{i} u_i^{t_i}(a_i',a_{-i})\geq \lambda \sum_i u_i^{t_i}(a')-\mu \sum_i u_i^{t_i}(a)$$
\end{defn}

The above class of games is a very general class of Bayesian Game and it 
is hard to believe that one can generalize smoothness to such a class.
However, a lot of the games in the literature satisfy an even stronger definition of smoothness.
For games that satisfy this stronger definition of smoothness we can generalize existing results
to incomplete information versions of the games.

\begin{defn}
 A Bayesian Game is universally $(\lambda,\mu)$-smooth iff $\forall t,w \in T$ and for all 
$a\in A(t)$, $b\in A(w)$:
$$\sum_i u_i^{w_i}(b_i,a_{-i})\geq \lambda \sum_i u_i^{w_i}(b)-\mu \sum_i u_i^{t_i}(a)$$
Since $b_i\in A_i(w_i)$ observe that the first term is also well defined.
\end{defn}

Universal smoothness is a more restrictive notion than smoothness, in the sense that if a Bayesian
Game is universally smooth then it is also smooth. This follows from the fact that if we take 
the definition of universal smoothness restricted only when $t=w$ then we get the smoothness definition. 
In addition the two definitions are equivalent to the smoothness of \cite{Roughgarden2009} for 
complete information games. 

Though the above definition seems restrictive enough in the sections that follow we will show
that most routing games studied in the literature are actually universally $(\lambda,\mu)$-smooth
and therefore the bounds known for the complete information carry over to some natural incomplete
information versions. 

\begin{theorem} 
 If a Bayesian Game with Variable Strategy Space is universally $(\lambda,\mu)$-smooth then
it has Bayes-Nash PoA at most $(1+\mu)/\lambda$.
\end{theorem}
\begin{proof}
Using similar reasoning as in Theorem \ref{thm:main} we can arrive at the conclusion that:
\begin{align*}
\E_t[SW^t(s(t))]~\geq~ & \E_{t} \E_{w} [\sum_i u_i^{w_i}(\opt_i(w), s_{-i}(t_{-i}))] \nonumber 
\end{align*}
Then we observe that $\opt(w)\in A(w)$ and $s(t)\in A(t)$. Thus applying the definition of
universal smoothness we get:
\begin{align*}
\E_t[SW^t(s(t))]~\geq~ & \E_t \E_w [\lambda \sum_i u_i^{w_i}(\opt(w))-\mu \sum_i u_i^{t_i}(s(t))]\\
~=~& \lambda \E_w[SW^w(\opt(w))]-\mu \E_t[SW^t(s(t))]
\end{align*}
which gives the theorem.
\end{proof}

\subsection{Weighted Congestion Games with Probabilistic Demands}
In this section we examine how our analysis applies to incomplete information versions of routing games. 
The games that we study in this section are cost minimization games and hence we will use the variants
of our theorems and notation so far adapted to cost minimization games. 

We first describe the complete information game. We consider unsplittable atomic selfish routing 
games where each player has demand $w_i$ that he needs to send from a node $s_i$ to a node $t_i$ over a graph $G$.
Let $\mathcal{P}_i$ be the set of paths from $s_i$ to $t_i$. The strategy of a player is to choose a path $p_i\in \mathcal{P}_i$.
Each edge $e$ of the graph has some delay function $l_e(x_e)$ where $x_e$ denotes
the total congestion of edge $e$: $x_e=\sum_{i: e\in p_i}w_i$, which is assumed to 
be monotone non-decreasing. Given a strategy profile $p=(p_i)_{i\in n}$
the cost of a player is $c_i(p)=\sum_{e \in p_i} w_i l_e(x_e)$.

In the literature so far only the case where $w_i$ are common knowledge has been studied. This is a very strong informational
assumption. Instead it is more natural consider the case where the $w_i$ are private information
and only a distribution on them is common knowledge. Thus the type a player in our game is his
weight $w_i$. In fact to make the game comply with our definition of Bayesian Games we will assume that the strategy of a player is a pair $(r_i,p_i)$ of a rate $r_i$ and path $p_i$. In addition given a type $w_i$ a player's
action space is $A_i(w_i)=\{(w_i,p_i): p_i\in \mathcal{P}_i\}$, i.e. we constraint the player to have to
route his whole demand. Given the above small alteration in the definition of the game it is now
easy to see that the cost of a player depends only on the strategies of the other players and 
not on their types: $\forall a_i=(r_i,p_i)\in A_i(t_i), \forall a_{-i}=(r_{-i},p_{-i})\in A_{-i}: c_i^{t_i}(a_i,a_{-i})=\sum_{e\in p_i} r_i l_e(x_e(a))$ where $x_e(a)=\sum_{k: e\in p_k} r_k$ and it's undefined for $a_i\notin A_i(t_i)$. 
Hence, if we prove that the latter Bayesian Game is universally $(\lambda,\mu)$-smooth,
this would imply a Bayes-Nash PoA of $\lambda/(1-\mu)$.

Very recently (Bhawalkar et al \cite{Bhawalkar2010}) showed that 
weighted congestion games are smooth games and therefore smoothness arguments provide tight 
results for the Price of Anarchy. Our analysis shows that one can extend these upper bounds to
incomplete information too. Moreover, since complete information is a special case of incomplete 
information where priors are singleton distributions, the bayes-nash price of anarchy analysis 
will still be tight. Moreover, this shows a collapse of efficiency between complete and incomplete
information and shows that knowing more doens't necessarily improve the efficiency guarrantee's in this
types of games. 

Most of the literature on weighted congestion games uses the following fact:
if for the class of delay functions $\C$ allowed we have that: 
$\forall x,x^*\in \mathbb{R}^+: x^*l_e(x+x^*)\geq \lambda x^*l_e(x^*)+\mu xl_e(x)$ then 
weighted congestion games with delays in class $\C$ are $(\lambda,\mu)$-smooth.
We will actually show that if the delay functions satisfy the above property then the
Bayesian Game is universally $(\lambda,\mu)$-smooth. 

\begin{lemma}If for any delay function $l_e()$ in the class of delay functions $\C$
allowed we have that: $\forall x,x^*\in \mathbb{R}^+: x^*l_e(x+x^*)\geq \lambda x^*l_e(x^*)+\mu xl_e(x)$
then the resulting class of Bayesian Unsplittable Selfish Routing Games with Probabilistic 
Demands is universally $(\lambda,\mu)$-smooth.
\end{lemma}
\begin{proof}
Let $w$, $t$ be two type profiles. Let $a=(w,p)\in A(w)$ and $b=(t,p')\in A(t)$. Let 
$x_e(a)=\sum_{i:e\in p_i} w_i$ and $x_e(b)=\sum_{i:e\in p_i'}t_i$. Then:
\begin{align*}
\sum_i c_i^{t_i}(b_i,a_{-i})\leq & \sum_i \sum_{e\in p_i'}t_i l_e(t_i+x_e(a))
\leq  \sum_i \sum_{e\in p_i'}t_i l_e(x_e(b)+x_e(a))\\
= & \sum_e x_e(b)l_e(x_e(b)+x_e(a))\\
\leq & \lambda \sum_e x_e(b)l_e(x_e(b))+\mu \sum_e x_e(a) l_e(x_e(a))\\
= & \lambda \sum_i \sum_{e\in p_i'} t_i l_e(x_e(b))+ \mu \sum_e \sum_{e\in p_i} w_i l_e(x_e(a))\\
= & \lambda \sum_i c_i^{t_i}(b) + \mu \sum_i c_i^{w_i}(a)
\end{align*}
which is exactly the universal smoothness definition.
\end{proof}

\subsection{Bayesian Effort Games}
In this section we study what our analysis imply for incomplete information versions of 
the following class of effort games \cite{Bachrach2011}: There is a set of players 
$[n]$ and a set of project $[m]$. Each player has some budget of effort $B_i$ which he can split among the projects. Each project $j$ has some value that is a non-decreasing concave function $V_j()$ of the weighted sum of efforts $\sum_{i\in N} a_{ij} x_{ij}$ where $a_{ij}$ is some ability factor of player $i$ in project $j$ (we assume $V(0)=0$). 
The value of a project is then split among
the participants and each participant receives a share proportional to his weighted input $a_{ij}x_{ij}$.
Such games where shown in \cite{Bachrach2011} to be Valid Utility Games \cite{Vetta2002} and hence $(1,1)$-smooth.

We will consider the natural incomplete information version of these games where each players
ability vector $a_i=(a_{ij})_{j\in [m]}$ and the budget $B_i$ are private information, each ability vector
is drawn from some distribution $F_i$ on $\R^{m+1}_+$. The $F_i$ are independent. To 
adapt it in our variable strategy space model we will assume that the strategy of a player
is a pair $(\tilde{a}_i,x_i)$ where $\tilde{a}_i$ is the declared ability vector and
$x_i$ is the vector of efforts of player $i$. We will constraint the strategy space
such that given an ability vector $a_i$ the player has to declare his true ability vector: 
$A_i(a_i,B_i)=\{(a_i,x_i): x_i\in \R^m, \sum_j x_{ij}\leq B_i\}$.

We are able to show that these games are actually universally smooth games and thereby
the Bayes-Nash PoA of the above Bayesian Games will be at most $2$.

\begin{lemma} Bayesian Effort Market Games are universally $(1,1)$-smooth. \end{lemma}
\begin{proof}
The proof is an adaptation of the smoothness proof for valid utility games \cite{Vetta2002,Roughgarden2009}, but in the space of
real functions instead of set functions. In addition it is adapted to accommodate for different types
of players so as to show the stronger version of universal smoothness is satisfied. 

Let $s=(a,x)\in A(a,B)$ and $s'=(b,y)\in A(b,B')$. Then we have:
\begin{align*}
\sum_i u_i^{(b_i,B_i)}(s_i',s_{-i}) = & \sum_i \sum_j b_{ij} y_{ij} \frac{V_j(b_{ij}y_{ij}+a_{-i}\cdot x_{-i})}{b_{ij}y_{ij}+a_{-i}\cdot x_{-i}} 
\end{align*}
Now we use the fact that for a concave function $V_j()$ that satisfies $V_j(0)=0$, it holds that
$V_j(x)/x$ is a decreasing function. Hence: 
$$\frac{V_j(b_{ij}y_{ij}+a_{-i}\cdot x_{-i})}{b_{ij}y_{ij}+a_{-i}\cdot x_{-i}} \leq 
\frac{V_j(a_{-i}\cdot x_{-i})}{a_{-i}\cdot x_{-i}}\implies 
b_{ij} y_{ij} \frac{V_j(b_{ij}y_{ij}+a_{-i}\cdot x_{-i})}{b_{ij}y_{ij}+a_{-i}\cdot x_{-i}} 
\geq V_j(b_{ij}y_{ij}+a_{-i}\cdot x_{-i})-V_j(a_{-i}\cdot x_{-i})$$
Thus:
\begin{align*}
\sum_i u_i^{(b_i,B_i)}(s_i',s_{-i}) \geq \sum_i \sum_j (V_j(b_{ij}y_{ij}+a_{-i}\cdot x_{-i})-V_j(a_{-i}\cdot x_{-i}))
\end{align*}
In addition, since $V_j()$ is concave, increasing then for all $t_1>t_2$ and $y>0$: $V_j(y+t_1)-V_j(t_1)\leq V_j(y+t_2)-V_j(t_2)$. Combining we get:
\begin{align*}
\sum_i u_i^{(b_i,B_i)}(s_i',s_{-i}) \geq& \sum_j \sum_i (V_j(b_{ij}y_{ij}+a_{-i}\cdot x_{-i})-V_j(a_{-i}\cdot x_{-i}))\\
\geq & \sum_j \sum_i V_j\left( b_{ij}y_{ij}+a_{-i}\cdot x_{-i}+ a_{ij}x_{ij}+\sum_{k=1}^{i-1}b_{kj}y_{kj}\right)-V_j\left(a_{-i}\cdot x_{-i}+a_{ij}x_{ij}+\sum_{k=1}^{i-1}b_{kj}y_{kj}\right)\\
= & \sum_j \sum_i  V_j\left( \sum_{k=1}^i (b_{kj}y_{kj}+a_{kj}x_{kj})+ \sum_{k=i+1}^{n}a_{kj}y_{kj}\right)-V_j\left(\sum_{k=1}^{i-1} (b_{kj}y_{kj}+a_{kj}x_{kj})+ \sum_{k=i}^{n}a_{kj}y_{kj}\right)\\
= & \sum_j V_j\left(\sum_{k=1}^{n} (b_{kj}y_{kj}+a_{kj}x_{kj})\right)- V_j\left(\sum_{k=1}^{n} a_{kj}x_{kj}\right)\\
\geq & \sum_j V_j\left(\sum_{k=1}^{n} b_{kj}y_{kj}\right)- V_j\left(\sum_{k=1}^{n} a_{kj}x_{kj}\right)\\
= & SW^{b,B'}(s')-SW^{a,B}(s)
\end{align*}
\end{proof}

\begin{corollary}
The Bayes-Nash PoA of Bayesian Effort Games is at most $2$.
\end{corollary}

\bibliographystyle{abbrv}
\bibliography{bayesian_smoothness}

\end{document}